\documentclass[10pt]{article}

\usepackage{amsmath}
\usepackage{amssymb}
\usepackage{amsthm}
\usepackage{graphicx}
\usepackage{algorithmic}

\usepackage{cite}

\usepackage{color} 

\usepackage[labelfont=bf,labelsep=period,justification=raggedright]{caption}

\makeatletter
\renewcommand{\@biblabel}[1]{\quad#1.}
\makeatother

\date{}

\begin{document}

\begin{flushleft}
{\Large
\textbf{Solving the Rural Postman Problem using the Adleman-Lipton Model}
}
\\
Nicolaos Matsakis$^{1}$
\\
\bf{1} Department of Information and Computer Science, UC Irvine, CA, USA
\\
$\ast$ E-mail: nmatsaki@uci.edu
\end{flushleft}

\section*{Abstract}

In this survey we investigate the application of the Adleman-Lipton model on Rural Postman problem, which given an undirected graph $G=(V,E)$ with positive integer lengths on each of its edges and a subset $E^{'}\subseteq E$, asks whether there exists a hamiltonian circuit that includes each edge of $E^{'}$ and has total cost (sum of edge lengths) less or equal to a given integer B (we are allowed to use any edges of the set $E-E^{'}$, but we must use all edges of the set $E'$). The Rural Postman problem (RPP) is a very interesting NP-complete problem used, especially, in network optimization. RPP is actually a special case of the Route Inspection problem, where we need to traverse all edges of an undirected graph at a minimum total cost. As all NP-complete problems, it currently admits no efficient solution and if actually $P\neq NP$ as it is widely accepted to be, it cannot admit a polynomial time algorithm to solve it. The application of the Adleman-Lipton model on this problem, provides an efficient way to solve RPP, as it is the fact for many other hard problems on which the Adleman-Lipton model has been applied. In this survey, we provide a polynomial algorithm based on the Lipton-Adleman model, which solves the RPP in $\mathcal{O}(n^{2})$ time, where n refers to the input of the problem. 

\section*{Introduction}

The breakthrough point a potential biological computer introduces, is parallelism. To make this more clear, we state that a few milliliters of water contain around $10^{22}$ molecules. Since biological computers work at the molecular level, we can imagine how enormous the parallelism of a biological computer would be, by only refering to this example. However, a potential biological computer would surely lack single process step time, compared to an ordinary computer, since it would be able to perform only a small fraction of a single operation in 1 second \cite{Lipton95usingdna}. But the difference in parallelism is so great, that outerperforms easily the single process step time. A serious disadvantage of this technique is that errors are involved into biological experiments, giving smaller yields. However, this can be surpassed by applying again the experiment many times and thus verifying the results, with a probability which tends to 1.

First, Adleman proposed a way to apply biological experiments in order to efficiently solve the Hamiltonian Path problem\cite{Adleman94molecularcomputation} and, afterwards, Lipton applied Adleman's technique on the SAT problem \cite{Lipton95usingdna}. It must be understood that the Adleman-Lipton model, does not provide an efficient way to solve any instance of NP-complete problems, rather than relatively medium-sized instances, on which however a contemporary computer would be impractical to use, since it would need unacceptably large exponential time to solve them. There have appeared various applications of this model on NP-complete problems since 1994. Apart from the Hamiltonian Path and the SAT Problem, there have been molecular computing solutions for the 3-Colouring, the Independent Set\cite{DBLP:journals/ppl/ChangGW05}, the Knapsack, the Subgraph Isomorphism, the Maximum Clique, the Shortest Common Superstring, the Set Splitting\cite{Chang:2003:SSP:1761566.1761593}, the Bounded Post Correspondence, the Traveling Salesman\cite{Shin99solvingtraveling}, the Monochromatic Triangle problems etc. A very good study on solving various NP-complete problems using biological experiments can be found in \cite{Katsany}.  

\section*{The DNA model of Computation}
Following \cite{xiao-li} we state the model on which our algorithm will be based. First of all, DNA is a polymer which comprises of monomers, called nucleotides. Nucleotides are detected only by their bases, which are Thymine (T), Guanine (G), Cytosine (C) and Adenine (A). Two strands of DNA can form a double strand (dsDNA) if each respective basis is the Watson-Crick complement of one other, which means that A should match T and C should match G. The length of a single-stranded DNA is the total number of its nucleotides and if this number is y, then the strand is called a y-mer. If we refer to a double-stranded DNA, then its length is measured by the base pairs (bp) it has.

A test tube is a set of molecules of DNA, in other words a multiset of finite strings over the alphabet defined by the 4 characters of the corresponding nucleotides. The following operations can be performed using test tubes \cite{dna,xiao-li}:
\\

0. Input(T): Creates the initial multiset, e.g. all the possible combinations of connected strands that test tube T comprises.

1. Merge$(T_{1},T_{2})$: Given 2 test tubes $T_{1}$ and $T_{2}$, we store the union of them in the first tube and leave the second tube completely empty of molecules.

2. Copy$(T_{1},T_{2})$: Given a test tube $T_{1}$, we copy its molecules and produce an identical tube $T_{2}$. 

3. Detect($T_{1}$): Given a test tube, we output 'YES' if the test tube contains at least one strand, whatever strand this would be. It is achieved by gel electrophoresis.

4. Separation($T_{1},X,T_{2}$): Given a test tube $T_{1}$ and a set of strings X, this action removes all single strands containing as a substring a string in X and produces a test tube $T_{2}$ with the strands we removed. 

5. Selection($T_{1},L,T_{2}$): Given a test tube $T_{1}$ and an integer L, this action removes all strands with length equal to L from $T_{1}$ and gives a test tube $T_{2}$ with the strands we removed. It is achieved by gel electrophoresis.

6. Annealing(T): Given a test tube T it produces all feasible double strands according to the predescribed Watson-Crick rule. 

7. Denaturation(T): Given a test tube T it dissociates each double strand in it, into 2 single strands. This action is the opposite of Annealing and can be achieved by providing the solution with a high temperature. It is, also, called Melting.

8. Discard(T): It discards tube T.

9. Append(T,Z): Given a test tube T and a short sDNA strand Z, it appends Z onto the end of every strand in tube T.
\\

The important is that all of the above operations can be implemented with a constant number of steps in a biological experiment, in other words each of the actions takes $\mathcal{O}(1)$ time, something that does not hold for an ordinary computer.

\section*{Using the Adleman-Lipton model to solve the Rural Postman problem}
\subsection*{The Rural Postman problem}

The Rural Postman Problem is defined in the following way\cite{garey}:

Our input is an undirected graph G=(V,E), lengths $l(e)\in Z^{+}_{0}$ $\forall e\in E$, $E'\subseteq E$ and a positive integer B.

The question is whether there exists a hamiltonian cycle in G that includes each edge in E' and which has total length no more than B. 

We can easily realize that for $E^{'}=\emptyset$ we get the Traveling Salesman Problem. For example, in figure 1 we can see a hamiltonian circuit over a graph G (dotted edges do not belong to the circuit) where all edges of $E^{'}$ belong to the circuit (we assume that $E^{'}$ here comprises only of the 2 edges that connect the vertex on the far left to the rest of the circuit) and where the rest of the edges (both taken in the circuit or are the dotted ones) belong to $E-E^{'}$.

\begin{figure}[!ht]
\begin{center}
\includegraphics[scale=0.30,trim = 20mm 100mm 35mm 5mm, clip]{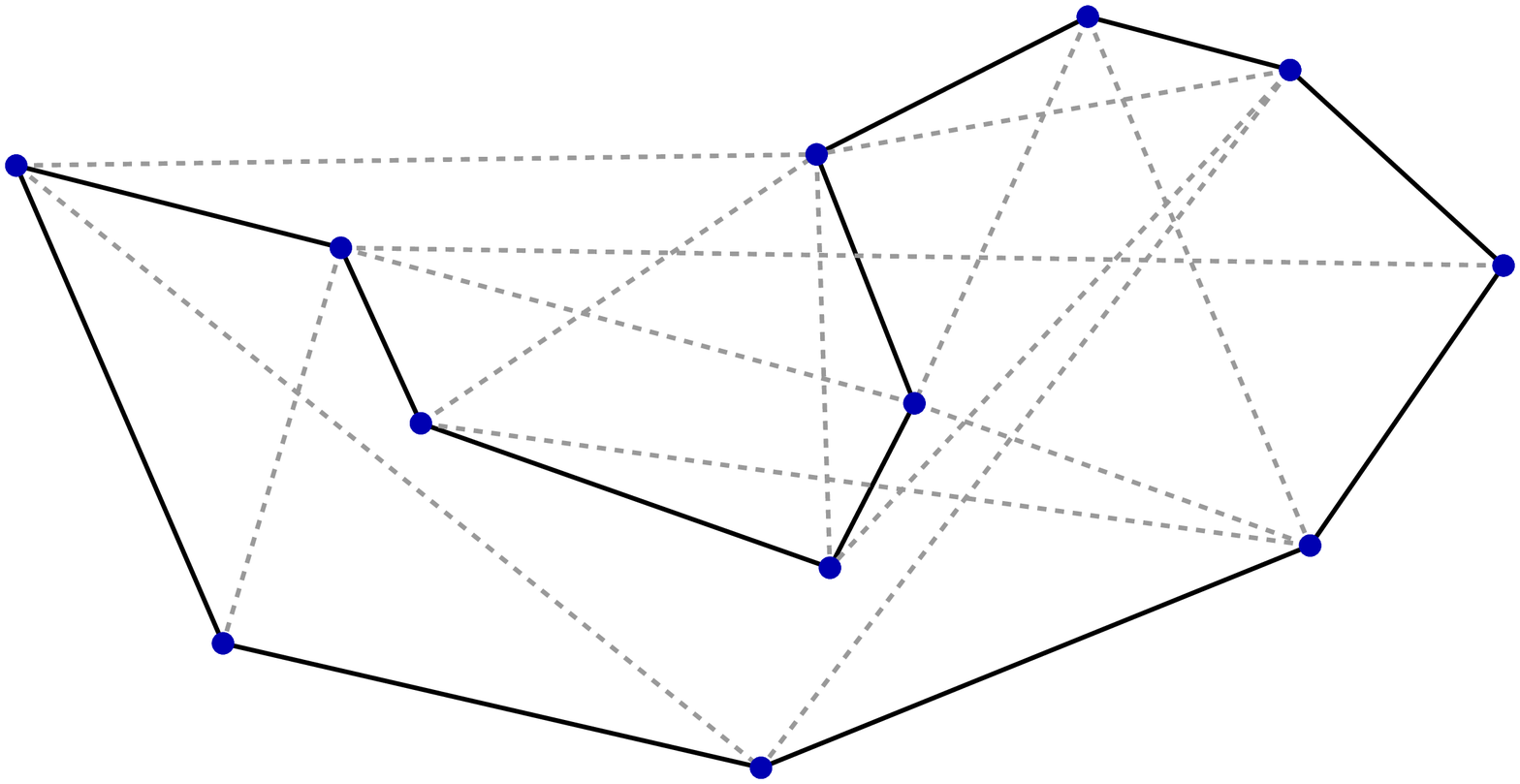}
\end{center}
\caption{
{\bf A Rural Postman circuit in a graph.} $E^{'}$ comprises of the 2 edges that connect the vertex on the far left to the rest of the circuit.} 
\label{Figure 2}
\end{figure}




\subsection*{Application of the Adleman-Lipton model on the RPP}

The algorithm is based on 4 steps: The first one is to produce all possible cycles of length equal to the number of vertices $v$ of the graph. The second one, is to obtain all cycles that use all edges from set $E^{'}$ as RPP constrains us to do. The third step checks whether the remaining cycles of the graph are actual circuits visiting each node once or not. If the answer to the previous step is positive then we must proceed to the last step which is the actual measurement of the total cost (total length) of the circuit.

The first thing we must do is to produce all possible instances of the problem. We assign a number from 1 to $v=|V|$ regarding each vertex of the graph, so we start from a random vertex assigning number 1 to it, we continue to one of its neighbours assigning number 2 and by traversing (using the DFS algorithm) the rest of the graph we assign all numbers until $v$. The next step is to assign 2 ordered pairs of values to each one of the edges. So, each edge will be assigned 2 pairs of integer values (ranged from 1 to $v$) where each one of the two integers corresponds to one vertex, providing the pair of vertices that this specific edge connects. As an example, the edge that connects the vertices 5 and 9, is assigned the ordered pair (5,9) and since we refer to undirected graphs we must assign to this edge the ordered pair (9,5) also, for the correctness of further analysis. 

Now we must encode the vertices and edges of the graph as DNA strands. So, for each one of the $v$ vertices we assign a different 20-mer single DNA strand to it. We can view this as two 10-mers connected together. Obviously, since there can be $4^{20}$ different 20-mer strands, there is absolutely no restriction doing that for a practical instance of the problem. The only thing we should notice is not to have a common first or second 10-mer of any of these strands with another first or second 10-mer of another strand. So, we require that each vertex is assigned to a 20-mer strand, where the first and second 10-mers of each strand encoding are identical to each other and, furthermore, different to each one of the 10-mers which encode the rest of the vertices (forming as pairs, the corresponding 20-mers). This decreases the amount of vertices we can encode to $4^{10}=2^{20}\approx 10^{6}$, which is still greater than the number of vertices we would need to encode, for any practical instance. An example of a vertex encoding is the following: 'ACTGAATGTAACTGAATGTA'. We can see that the first 10-mer is identical to the second 10-mer. We call this set of sDNA strands P. 

After assigning the single DNA strands (sDNA) to the vertices, we assign sDNA strands to the edges (which are at most $O(n^{2})$ for any simple graph) according to the following rule: For each vertex, we think of its 20-mer strands as strings (of 20 chars) where $\Sigma=\{A,C,G,T\}$, in other words as two 10-mers connected together, as we indicated before. We take all of the $v$ different 10-mers we have with this way (please note again that a vertex encoding comprises of two identical 10-mers) and create new 20-mer strands, by connecting each one of the 10-mer strands with each one of the rest of them that correspond to neighbour vertices, taking new 20-mers. The difference, now, is that the first and the second 10-mers of each strand are not identical as in the case of vertices. Also, each edge is encoded with 2 ways as we indicated before, e.g. with the 2 combinations of the two 10-mer encodings of the vertices that it connects. However, since we, clearly, head on creating double strands which will form paths inside the graph, we should obtain the Watson-Crick complements of those specific 20-mers in order to bind them with the sDNA strands of vertices. So, we take the complements of each one of these strands that we just created for the case of the edges. For example, an edge between vertices $'$ACTGAATGTAACTGAATGTA$'$ and $'$AGATTCACTGAGATTCACTG$'$ is encoded as the strand $'\overline{ACTGAATGTAAGATTCACTG}'$ and, also, as the strand $'\overline{AGATTCACTGACTGAATGTA}'$. We call this set of sDNA strands Q. However we do not place all of the edge formulations in Q; we pick one edge, randomly, from $E'$ (we call it $e'$) and leave it aside, excluding it from Q. 

Finally, we take the two 10-mers that encode $e'$ (as done before for the rest of edges), but we do not connect them to form two 20-mers. We leave them as two 10-mers and call this set R.

Now, following the procedure of \cite{xiao-li} we get the following, assuming that we have already created the initial multisets of P and Q where the strands of P have length $20v$ and the strands of Q length $20*(|v|-1)$, since they correspond to hamiltonian circuits and the set of edges lacks edge $e'$.
\\
\begin{algorithmic}
\STATE Merge (P,Q);
\STATE Merge (P,R);
\STATE Annealing (P);
\STATE Denaturation (P);
\STATE Selection(P,20$v$,R);
\end{algorithmic}
\begin{verbatim}\end{verbatim}

Please note that the Merge operation leaves the second tube always empty. Now, tube R contains, among others, all possible cycles of length $v$ in G. In other words, it contains a vast amount of possible paths in the undirected graph. A cycle of length $v$ does not have to be a hamiltonian circuit since it may revisit nodes and therefore avoid visiting some of the remaining vertices of the graph, even though it may give the desirable length of $v$ in the graph (20$v$ for the strand encodings). The other thing we are sure about is that edge $e'$ is a part of all those cycles. So, we have to start pruning the result and end up to the question that the problem sets to us.

\begin{figure}[!ht]
\begin{center}
\includegraphics[scale=0.30,trim = 20mm 180mm 35mm 5mm, clip]{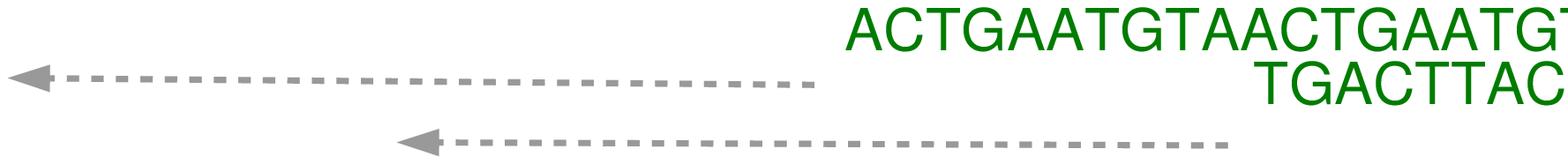}
\end{center}
\caption{
{\bf Regarding our previous example, this is the corresponding (partial) double strand associated with the $vertex_{1}-edge-vertex_{2}$ part of the path.} 
The 2 vertices (upper strand) bind with the connecting edge (lower strand) to a double strand. Each upper nucleotide is the complement of the one exactly below it. Of course the double strand expands to the left and right with the only difference that the lower sDNA strand ends 10 characters before the upper one in both directions and there is where the edge $e'$ is used, to fill in these 2 gaps}
\label{Figure 2}
\end{figure}


The first action is to take all remaining edges in E' (we indicate their strand encodings as $E^{'}_{i}$ where i is an integer from 1 to $|E^{'}|-1$) and subtract the strands that contain only all of those edges from the tube. If the cardinality of $E^{'}$ is m, then we must check for the remaining m-1 edges in $E^{'}$, after having enumerated them. So we create empty tubes $L_{1}$ to $L_{m-1}$ and apply the following procedure:
\\
\begin{algorithmic}
\STATE Copy(R,$L_{1}$);

\FOR{$i = 1$ \TO m-1} 
\STATE Separation ($L_{i},E^{'}_{i},L_{i+1}$);
\ENDFOR
\end{algorithmic}
\begin{verbatim} \end{verbatim}

By doing that for every strand encoding $E^{'}_{i}$ consequently (except for $e^{'}$ which we are sure that it is included), we take all the sDNA strands containing all edges from $E^{'}$ and a possible number of edges from $E-E'$. The resulting tube is $L_{m}$.  Also, we indicate that we removed all the sDNA strands that resulted from denaturation and which corresponded to the sDNA strands encoding vertices. The reason is that those corresponded to complement strands, according to the Watson-Crick rule, so they were separated.

Now, tube $L_{m}$ has the encoded strands of all cycles using \textit{for sure} all edges from E' and possibly edges from the set $E-E'$. However, we have not distinguished yet cycles (in general) from circuits, which is what we want to have in the end.


\begin{figure}[!ht]
\begin{center}
\includegraphics[scale=0.30,trim = 20mm 80mm 35mm 5mm, clip]{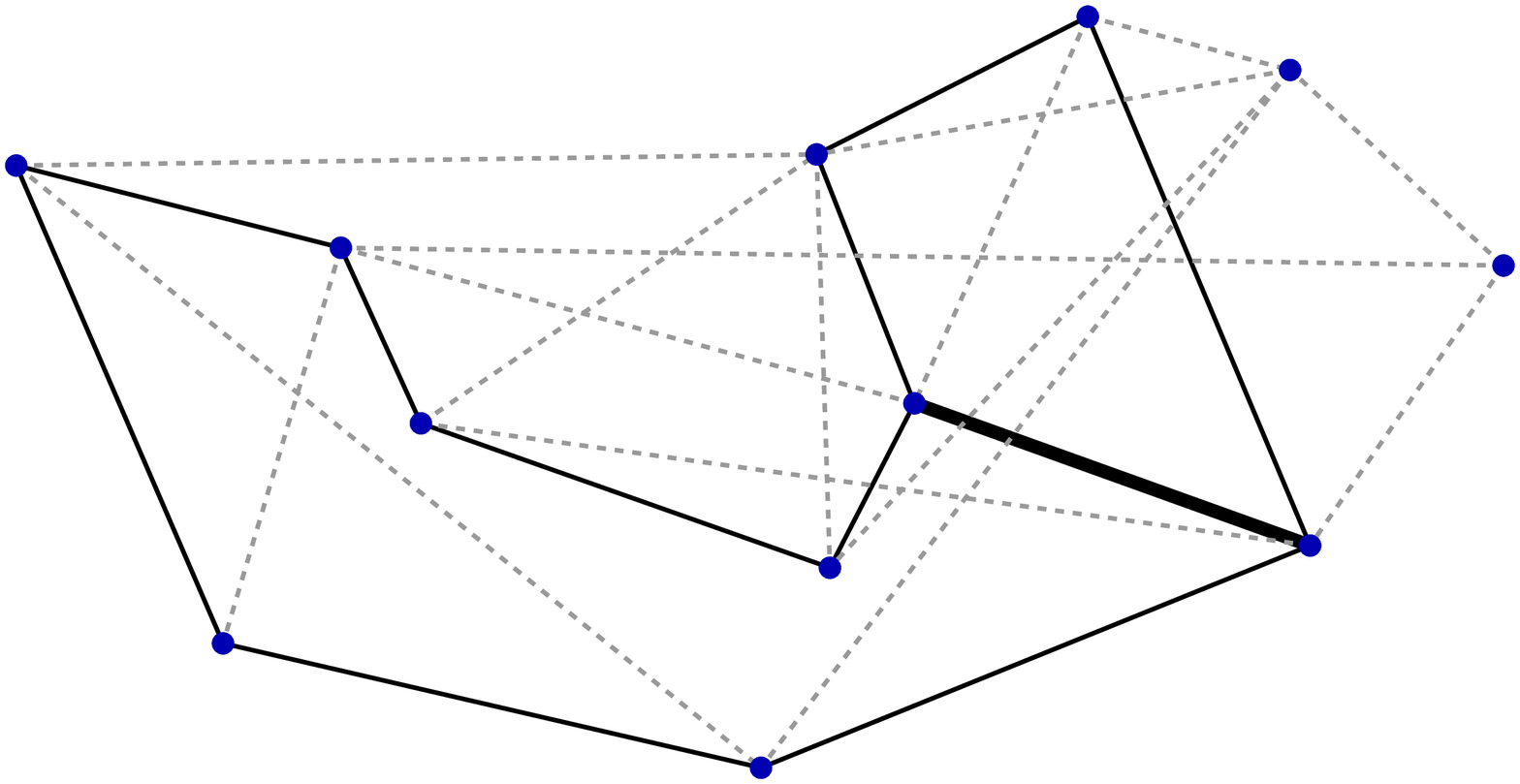}
\end{center}
\caption{
{\bf A cycle in a graph which is not a hamiltonian circuit, though it has a hamiltonian circuit's length}. The edge in bold is passed through twice back and forth consequently, leaving aside the 2 upper right vertices which should have been visited instead.}
\label{Figure 2}
\end{figure}

So, we must check whether all vertex corresponding strands are still there. We indicate as $\overline{V_{i}}$ the Watson-Crick complement of the i-th vertex encoding. In order to leave aside those strands from $L_{m}$ which do not encode hamiltonian circuits, we must check for each vertex of G whether its encoding is part of the strand. If, for example, two vertices are missing, then we have 2 repetitions of already encoded vertices, since we are sure for the total length which is 20*$v$. So, for every vertex in G, we apply the following procedure, where N and Temp are new empty tubes:
\\

\begin{algorithmic}

\FOR{$i = 1$ \TO $v$} 
\STATE Separation($L_{m},{\overline{V_{i}}}$,Temp);
 
\IF{Detect(Temp)} 
\STATE {Merge(N,Temp); $i\gets i+1$} 
\ELSE 
\STATE {Reply('NO'); Exit}
\ENDIF 
\ENDFOR

\STATE Reply('YES')
\STATE Exit

\end{algorithmic}
\begin{verbatim}\end{verbatim}
This procedure takes, clearly, $\mathcal{O}(n)$ time.

If the reply is 'NO', then there is no need to search for the total length, since the graph does not contain a hamiltonian circuit using among others, each one of the obligatory edges. If the answer is 'YES', then we should move to the following step which is the measure of the total length. 

We are asked whether there is a circuit with total cost (sum of edge lengths) less or equal to the integer B. We can restate this question, by asking whether, in the formed circuit, the edges from $E-E^{'}$ have total cost less or equal to $c=B-\sum_{i=1}^{i=|E'|}c_{i}$ where $c_{i}$ refers to the length of the i-th edge in set $E'$ (assuming the previous enumeration of them, on the end of which we place edge $e^{'}$), since we know that all edges from $E'$ have been used. So our concentration should focus on the rest of the edges used in the circuit which we know how many should be $(v-|E^{'}|)$ but we have no idea which edges are. This will be the final step of our algorithm.

We append to each strand in tube N a (not used before) 20-mer sDNA strand encoding character '@':
\\

Append(N,@)
\\

In order to encode a number using DNA sequences we will use the number of nucleotides of a random strand. So, a y-mer encodes integer y. We construct a procedure, that will be repeatedly removing edges from $E-E^{'}$ adding their length y as a y-mer after '@'. When we finish, we can count the number of total nucleotides after '@' and if this number is less or equal to c, the answer is 'YES', otherwise it is 'NO'. However, since the Append procedure works only for relatively small strand lengths (less than 20 nucleotides), we should apply this procedure repeatedly (l div 20) times for greater lengths l, appending each time a 20-mer on the end to the strand and finally append a (l mod 20)-mer to the end of the current strand. Clearly, the total nucleotide augmentation adds up to l. This way, we can formalize the append of any edge length.

First, we enumerate from i=1 to $k=|E|-|E^{'}|$ all free edges and take their strand encodings as before, so we get each $E_{i}$. L is denoted as the current length of the strands in N (so $L=20*|v|+20$ because of the append of '@') and $F$, $Z_{i}$ ($i\in\{1,..,k\}$) are tubes that we will use. The (\textit{length}($E_{i}$)-mer) is the random polymer of this specific length that we associate with the length of the corresponding edge. So, for example, an edge with length 17 is associated here with a random 17-mer. Then, we get the following, where we assume that each edge length is no more than 20:
\\

\begin{algorithmic}
\FOR {$i=1$ \TO $k$}
 \STATE {Separation(N,$E_{i}$,$Z_{i}$)};
  \IF {Detect($Z_{i}$)}
       \STATE {Append($Z_{i}$,(length($E_{i}$)-mer)); Merge(N, $Z_{i}$)}
  \ENDIF
\ENDFOR

\FOR {$i=0$ \TO c} 
\STATE{Selection($N$,L+c-i, $F$)};
\IF {Detect(F)}
\STATE{Reply ('YES'); Exit}
\ENDIF
\ENDFOR

\STATE{Reply('NO')}; 

\STATE{Exit}
\end{algorithmic}
\begin{verbatim} \end{verbatim}

\newtheorem{theorem}{Theorem}
\begin{theorem}
The above algorithm solves the decision problem of the Rural Postman in $\mathcal{O}(n^{2})$ time.
\end{theorem}

\begin{proof}
A graph may have up to $\mathcal{O}(n^{2})$ edges. Since we apply the first 'for' procedure k times and $k=\mathcal{O}(n^{2})$, the theorem is easily derived, assuming B is fixed in the input of the problem.
\end{proof}

Now, if an edge length exceeds 20 then we should apply the \textit{Append}($Z_{i}$,20-mer) procedure $\left\lfloor length/20 \right\rfloor$ times (where 20-mer is a random 20-mer) and then apply \textit{Append}($Z_{i}$,((\textit{length mod 20})-mer)) once, instead of what we actually do above. However, for providing a clearer result we mention this externally.

Finally, we state the whole algorithm as a whole, for any edge length:
\\

\begin{algorithmic}
\STATE Merge (P,Q);
\STATE Merge (P,R);
\STATE Annealing (P);
\STATE Denaturation (P);
\STATE Selection(P,20$v$,R);
\STATE Copy(R,$L_{1}$);
\FOR{$i = 1$ \TO m-1} 
\STATE Separation ($L_{i},E^{'}_{i},L_{i+1}$);
\ENDFOR

\FOR{$i = 1$ \TO $v$} 
\STATE Separation($L_{m},{\overline{V_{i}}}$,Temp);
 
\IF{Detect(Temp)} 
\STATE {Merge(N,Temp); $i\gets i+1$} 
\ELSE 
\STATE {Reply('NO'); Exit}
\ENDIF 
\ENDFOR
\STATE Append(N,@)
\FOR {$i=1$ \TO $k$}
 \STATE {Separation(N,$E_{i}$,$Z_{i}$)};
  \IF {Detect($Z_{i}$)}
     \IF {length($E_{i})\leq 20$}
                   \STATE {Append($Z_{i}$,(length($E_{i}$)-mer)); Merge(N, $Z_{i}$)}
     \ELSE 
                   \FOR {$j=1$ \TO $\left\lfloor (length(E_{i})/ 20\right\rfloor$}
                      \STATE {Append($Z_{i}, 20-mer)$}
                   \ENDFOR
                   \STATE {Append($Z_{i},((length(E_{i}) \mod 20)-mer))$}
             
      \ENDIF
  \ENDIF
\ENDFOR

\FOR {$i=0$ \TO c} 
\STATE{Selection($N$,L+c-i, $F$)};
\IF {Detect(F)}
\STATE{Reply ('YES'); Exit}
\ENDIF
\ENDFOR

\STATE{Reply('NO')}; 

\STATE{Exit}

\end{algorithmic}


\section*{Conclusions}

There have been many applications of the Adleman-Lipton model on various NP-complete problems. Though the effectiveness of the Adleman-Lipton model has been questioned during the last years, it is, surely, at least a method of theoretical importance which may have even further consequences in the area of Algorithms. Here, we presented a method of applying this model on the Rural Postman problem. Finally, we indicate that with a very easy modification of our algorithm (associating each edge with only one ordered pair of vertices and not with two as we did), we can apply this method to the directional variation of the Rural Postman problem which is, also, NP-complete \cite{garey}.

\section*{Acknowledgments}

\bibliography{sigproc}



\end{document}